\documentclass[12pt]{article}

\usepackage[T1]{fontenc}
\usepackage[utf8]{inputenc}
\usepackage[a4paper,margin=3cm]{geometry}
\usepackage{bbm}
\usepackage[colorlinks=true,urlcolor=blue,linkcolor=blue,citecolor=blue]{hyperref}
\usepackage{soul}
\usepackage{marvosym}
\usepackage{amsfonts}
\usepackage{xspace}
\usepackage{color}
\usepackage{booktabs}
\usepackage{xtab}
\usepackage{latexsym}
\usepackage[sc]{mathpazo}
\usepackage{xfrac}
\usepackage{mathtools}
\usepackage{amssymb}
\usepackage{comment}
\usepackage{tikz} %realizar dibujos
\usetikzlibrary{calc}
\usepackage{tikz-3dplot} %dibujos en 3D
\usepackage{subcaption} %entorno de subfiguras

\newcommand{\mfrac}[2]{\mbox{\footnotesize$\displaystyle\frac{#1}{#2}$}}
\newcommand{\diff}{\mathrm{d}}
\newcommand{\ds}{\diff s}
\newcommand{\dt}{\diff t}

\newcommand{\dd}[2][{}]{\frac{\!\diff#1}{\!\diff#2}}

\newcommand{\dds}{\dd{s}}

\makeatletter
\newcommand{\Jac}{\mathop{\operator@font Jac}}
\newcommand{\rang}{\mathop{\operator@font rang}}
\newcommand{\sol}{\mathop{\operator@font sol}}
\newcommand{\argmin}{\mathop{\operator@font argmin}}
\makeatother
\newcommand{\tendsto}{\rightarrow}

\newcommand{\To}{\longrightarrow}
\newcommand{\quand}{\quad\text{and}\quad}

\newcommand{\coloneq}{\vcentcolon\mathrel{\mkern -1.2mu}\mathrel{=}}
\newcommand{\spacetime}{(M,\Omega,g,\nabla)}

\DeclareMathOperator{\pt}{\partial_t}

\DeclareMathOperator{\An}{\text{An}\Omega}

\DeclareMathOperator{\grad}{grad}
\DeclareMathOperator{\rot}{Rot}
\DeclareMathOperator{\tor}{Tor}

\usepackage{amsmath,amssymb,amsthm,mathrsfs}

\allowdisplaybreaks[2]

\usepackage[english]{babel}

\newtheorem{defi}{Definition}[section]
\newtheorem{theor}[defi]{Theorem}
\newtheorem{pro}[defi]{Proposition}
\newtheorem{cor}[defi]{Corolary}
\newtheorem{lem}[defi]{Lemma}
\newtheorem{rem}[defi]{Remark}
\newtheorem{ex}[defi]{Example}

\def\R{\mathbb R}

\def\f{\mathcal F}

\def\x{\mathfrak X}

\def\U{\mathcal U}
\def\G{\mathcal G}

\let\olditemize\itemize
\def\itemize{\olditemize\itemsep=0pt }
\let\oldenumerate\enumerate
\def\enumerate{\oldenumerate\itemsep=0pt }

\title{On Twisted Spacetimes: a new class of Galilean cosmological models}
\author{%
    \textbf{Daniel de la Fuente} \\
    Departamento de Matemáticas\\
    Universidad de Oviedo, 33003 Gijón, Spain \\
    \texttt{fuentedaniel@uniovi.es}
    \and
    \textbf{Rafael María Rubio} \\
    Departamento de Matemáticas\\
    Universidad de Córdoba, 14071 Córdoba, Spain \\
    \texttt{rmrubio@uco.es}
    \and
    \textbf{Jose Torrente} \\
    Departamento de Matemáticas\\
    Universidad de Córdoba, 14071 Córdoba, Spain \\
    \texttt{jtorrente@uco.es}
    }
\date{February 2024}

\begin{document}

\maketitle

\begin{abstract} Within the generalized Newton-Cartan theory, Galilean Twisted spacetimes are introduced as dual models of the well-known relativistic twisted spacetimes. As a natural generalization, torqued vector fields in Galilean spacetimes are defined, showing that the local structure of a Galilean spacetime admitting a timelike torqued vector field is given by a Twisted spacetime. In addition, several results assuring the global splitting as Twisted spacetime are obtained. On the other hand, completeness of free falling observers is studied, as well as general geodesic completeness.
    
\end{abstract}

\section{Introduction}
The geometric formulation of Newton's theory of gravity began with E. Cartan in the early twentieth century \cite{Cartan1, Cartan2}. But it was only in the second half of the century that the Newton-Cartan theory was given significant attention by various authors.
In the 1960s Malament studied the Lorentzian and Galilean groups, proving that Newtonian gravity constitutes a limit of General Relativity and how free particles are represented by geodesic curves of certain connections in both theories.
Additionally, Trautman showed the existence of a symmetric connection in both theories \cite{Trautman}, arising due to the local nature of physical laws and simultaneously encoding the inertia principle and gravitational forces. Concretely, classical Newtonian gravitation is formulated as a covariant theory, showing that certain results previously considered characteristic or singular of the theory of Relativity are also shared by the (geometric) Newton-Cartan Gravitational Theory. In fact, Newtonian gravity is explained by the curvature of a connection in the spacetime, although in this case it no longer arises from any semi-Riemannian metric. Furthermore, in the geometric formulation of Newtonian Gravity Theory, the spacetime structure is dynamic in the sense that it participates in the development of physics rather than being a fixed backdrop (see \cite{DB} and references therein).

Later in that century, several steps were taken towards the generalization of this theory, such as the introduction of new Newtonian models satisfying the cosmological principle \cite{Muler1983}, obtaining the Galilean analogue of relativistic Robertson-Walker spacetimes. In \cite{Duval}, Newtonian theory achieves the status of a gauge theory and is shown to be a limit of the Einstein-Klein-Gordon theory. This view of the theory as a limit of General Relativity continued to develop during these years \cite{Leih}, concluding also that the results obtained from Newton-Cartan cosmology were similar on not too large scales to those derived from Einstein's theory \cite{Rod, Ruede}.

In the last few decades, the simplicity of certain Newtonian models in Newton-Cartan theory, in contrast to the computational complexity of General Relativity, has led to new applications of this non-relativistic theory in Cosmology \cite{Bra}, Hydrodynamics \cite{Geracie2}, quantum collapse \cite{Pen}, AdS/CFT correspondence \cite{Lin}, condensed matter systems \cite{Brau}, quantum Hall effect \cite{Gera} and other related phenomena. The many applications of Newton-Cartan theory have motivated a review of the geometric structures associated with generalized Newton-Cartan theory \cite{Bernal} and new extensions such as Newton-Cartan Supergravity \cite{Roel3} and non-relativistic Strings \cite{Roel, Berg}. Furthermore, new Newtonian models have been obtained, including Galilean Generalized Robertson-Walker spacetimes \cite{GGRW}, the Newtonian G{\"o}del spacetime \cite{Costa}, standard stationary Galilean spacetimes \cite{GFPR}, as well as the embedding of these geometries in 5-dimensional Lorentzian manifolds \cite{Beka,Minguzzi2006a, Minguzzi2007, Minguzzi2010}.
 
In this manuscript, a wide family of Galilean spacetimes is presented, the Galilean Twisted spacetimes (Sect. \ref{GT}), which constitute a generalization of the classical Galilean Robertson-Walker spacetimes (RW), and even of the Galilean Generalized Robertson-Walker spacetimes (GGRW). In fact, Definition \ref{defi:GT} states 
\begin{quote}
\upshape
Given a real interval $I\subseteq\R$, a Riemannian manifold $(F,h)$ and a positive smooth function $f\in C^\infty(I\times F)$, a Galilean Twisted (GT) spacetime is given by the tuple $(I\times F, \Omega=\dt, g=f^2 \, \pi_F^* h, \nabla)$ where $t\coloneq \pi_I$ and $\nabla$ is the only symmetric Galilean connection on $I\times F$ such that $\nabla_{\pt} \pt = 0$ and $\rot \pt = 0$.
\end{quote}
According to the cosmological principle, RW spacetimes model a spatially homogeneous and isotropic universe composed of matter (galaxies) subject only to the action of its own gravity. Thus, the ``absolute space at each moment in absolute time'' must be locally Euclidean with constant sectional curvature. Furthermore, as large-scale experimental measurements (redshift) indicate, the expansion is assumed to be homogeneous, i.e., the expansion rate is the same at all points at a fixed (absolute) time.  
However, while the hypothesis of spatial homogeneity and isotropy may be a reasonable first approximation to the large scale structure of the universe, it could not be appropriate when we consider a more accurate scale. This led to the introduction of the family of GGRWs spacetimes, where the absolute space is modelled by an arbitrary Riemannian manifold, which does not necessarily have to be of constant curvature. Finally, the Twisted spacetimes presented here assume a non-homogeneous expansion rate, which means that it depends not only on the moment in time but also on the point in space. In other words, the ``Hubble constant'' is no longer constant.

Twisted products have been studied in the relativistic setting \cite{Chen1979, Chen2017-warped, ManticaMolinari2017-survey, Soria2023}. They have been analyzed in relation to geometric models of fluids, resulting in a characterization in terms of the stress-energy tensor of an imperfect fluid \cite{ManticaMolinari2017-torseforming} and obtaining the properties of static perfect fluids in Twisted Lorentzian spacetimes \cite{Guler2023}. Similar results could be explored in the non-relativistic setting introduced here. 

A Galilean Twisted (GT) spacetime possesses an infinitesimal symmetry given by the existence of a timelike irrotational and (spatially) conformally Leibnizian field of observers.  Several geometrical properties, as well as physical interpretations of this family of spacetimes, are given in Sections \ref{sect4} and \ref{sect-phys}. We concretely focus on the completeness of its free falling observers. Our main result (Theorem \ref{theor:completeness}) asserts
\begin{quote} 
{\emph Let $(F,h)$ be a complete Riemannian manifold and let $f\in C^\infty(\R \times F)$ be a positive function. If $\inf f >0$ and $\inf (f_t/f)>-\infty$ along finite times, then the GT spacetime $(M=\R\times F, \Omega =\diff \pi_\R, g=f^2\, \pi_F^* h, \nabla)$ is geodesically complete.}
\end{quote}
From a physical perspective, this result is relevant since the completeness of the timelike trajectories is crucial in order to have a consistent theory (without observers or with observers that disappear or appear suddenly, it would not be possible to do Physics).

Section \ref{sect5} is devoted to the study of Galilean spacetimes admitting a torqued vector field (see Definition \ref{defi:torqued}), which is closely related to the previously introduced family of irrotational conformally Leibnizian vector fields \cite{GGRW}. The principal result in this section is Theorem \ref{theor:local structure torqued}, which gives
\begin{quote}
    Let $\spacetime$ be a Galilean spacetime.
    If it admits a torqued vector field $K\in \Gamma(TM)$, then for each $p\in M$, there exists an open neighborhood of p, $\U$, and a Galilean diffeomorphism $\Psi\colon N \to \U$, where $N$ is a GT spacetime. 
\end{quote}

Finally, Section \ref{sect6} is devoted to face some splitting type problems, i.e., under which geometrical assumptions a Galilean spacetime decomposes globally into a Galilean Twisted spacetime. Theorem \ref{theor:splitting global} states
\begin{quote}
\upshape
    Let $\spacetime$ be a simply connected Galilean spacetime.
    It admits a global decomposition as a GT spacetime if and only if there exists a Galilean torqued vector field $K\in \Gamma(TM)$ such that the associated field of observers, $Z\coloneq K/\Omega(K)$, is a uniform global generator.
\end{quote}
Moreover, if the leaves of the foliation induced by $\Omega$ are compact, we can ensure again a global splitting result (Theorem \ref{theor:splitting global compact}).

\section{Premilinaries}
Let $M$ be a $(n+1)$-dimensional smooth manifold. As usual, it is assumed connected, Hausdorff and paracompact.
A \textit{Leibnizian spacetime} \cite{Bernal} is a triad $(M,\Omega, g)$ where $(\Omega, g)$ is a {\it Leibnizian structure} defined by a global non-vanishing one-form on $M$, $\Omega\in \Lambda^1(M)$, and a positive definite metric $g$ on the vector bundle defined by the kernel of $\Omega$. Concretely, if $\An \coloneq \{ v\in TM \; : \; \Omega(v) = 0 \}$ is the n-distribution induced by $\Omega$, then $g$ is a smooth, bilinear, symmetric and positive definite 2-covariant tensor
\begin{equation*}
    g\colon \Gamma(\An) \times \Gamma(\An) \to C^\infty(M)\,,
\end{equation*}
where $\Gamma(\An)$ denote the subset of sections $V\in \Gamma(TM)$ such that $V_p\in \An$, $\forall p\in M$. 

Given a Leibnizian spacetime $(M,\Omega,g)$, the points $p\in M$ are called \textit{events}. If $p\in M$, $\Omega_p$ is known as the \textit{absolute clock} at $p$ and $\big(\text{An} (\Omega_p), g_p\big)$ is called the \textit{absolute space} at $p$ (see for example \cite{Bernal2}). Adopting relativistic terminology, a tangent vector $v\in T_pM$ is said to be spacelike (resp. timelike) if $\Omega_p(v)=0$ (resp. $\Omega_p(v)\ne 0$). Moreover, a timelike vector $v$ is called {\it future-pointing} (resp. {\it past-pointing}) if $\Omega_p(v)>0$ (resp. $\Omega_p(v)<0$). 

On the other hand, an {\it observer} is defined as a smooth curve $\gamma\colon I\subseteq \R \to M$, being $I$ an interval, such that its velocity is a normalized timelike future-pointing vector ({\it standard timelike unit} \cite{Bernal2}), i.e., $\Omega_{\gamma(s)}(\gamma'(s)) = 1,\, \forall s\in I$. The parameter $s\in I$ is the {\it proper time} of the observer $\gamma$. A \textit{field of observer} $Z\in \Gamma(TM)$ is a smooth vector field whose integral curves are observers, equivalently, $\Omega(Z)=1$. 

In this Leibnizian setting, the synchronizability of observers is intrinsic to the Leibnizian structure, in contrast to the  relativistic notions, which is relative to each field of observers. In fact, if $\Omega \wedge \diff \Omega = 0$ (i.e., the smooth distribution $\An$ is integrable), the spacetime is called {\it locally synchronizable}. In this case, Frobenius Theorem (see \cite{Warner1983}) states that there exists a foliation consisting of spacelike hypersurfaces. Moreover,  locally $\Omega = \lambda \diff T$, for certain smooth functions  $\lambda>0$ and $T$, and the hypersurfaces of the foliation are exactly $\{T=\text{constant}\}$. Thus, every observer can locally synchronize its proper time with the so-called "compromise time" $T$. If the absolute clock is closed, $\diff \Omega=0$, the Leibnizian spacetime is said to be \textit{proper time locally synchronizable}. Under this assumption, locally we can express $\Omega = \diff T$, for certain smooth function $T$, and the observers can synchronize their proper times up to a constant. In the case $\Omega=\diff T$, for $T\in C^\infty(M)$, we will suppose every observer is parameterized by the \textit{absolute time} $T$. Notice that, if $M$ is simply-connected and $\Omega$ is closed, we can ensure the existence of an absolute time function making use of the Poincaré Lemma (see, for instance, \cite{Spivak1}).

Let $(M,\Omega, g)$ and $(\widetilde{M}, \widetilde{\Omega}, \widetilde{g})$ be two Leibnizian spacetimes. A diffeomorphism $\lambda\colon M \to \widetilde{M}$ is called \textit{Leibnizian} if it preserves the absolute clock and space, that is, $\lambda^* \widetilde{\Omega} = \Omega$ and $\lambda^*\widetilde{g} = g$. 

To give complete physical meaning to a Leibnizian spacetime, the principle of inertia has to be included via an affine connection $\nabla$. It is necessary to consider $\nabla$ compatible with the absolute clock $\Omega$ and the absolute space $(\An, g)$, i.e., 
\begin{enumerate}
    \item[(1)] $\nabla \Omega = 0$, equivalently, $X(\Omega(Y)) = \Omega(\nabla_X Y)$, for all $X, Y \in \Gamma(TM)$,
    \item[(2)] $\nabla g= 0$, that is, $Z(g(V,W)) = g(\nabla_Z V, W) + g(V, \nabla_Z W)$, for all $Z\in \Gamma(TM)$ and $V,W\in \Gamma(\An)$.
\end{enumerate}

Note that, thanks to condition (1), the right-hand side of item (2) is well defined.
Any such $\nabla$ is called a \textit{Galilean connection}. Consequently, the Leibnizian spacetime does not determine a canonical Galilean connection. The tuple $\spacetime$ is known as a \textit{Galilean spacetime}. Recall that if its torsion tensor vanishes, a connection is said to be \textit{symmetric}, i.e.,
$$\text{Tor}(X,Y) = \nabla_X Y - \nabla_Y X - [X,Y] = 0, \quad \forall X, Y \in \Gamma(TM).$$

From a physical viewpoint, it is often convenient to choose a symmetric connection because it is determined by the associated geodesics \cite[Add. 1, Ch. 6]{Spivak2}. Observe that the closeness of the absolute clock and the torsion tensor are related \cite[Lemma 13]{Bernal}
\begin{equation}\label{eq:torsion and clock}
    \Omega \circ \tor = \diff \Omega\,.
\end{equation}
From now on, we will only consider symmetric Galilean connections, then using equation (\ref{eq:torsion and clock}), $\diff \Omega = 0$.

Given two Galilean spacetimes $\spacetime$ and $(\widetilde{M}, \widetilde{\Omega}, \widetilde{g}, \widetilde{\nabla})$, a diffeomorphism $\lambda\colon M \to \widetilde{M}$ is called {\it Galilean} if it is Leibnizian and additionally $\lambda^* \widetilde{\nabla} = \nabla$, i.e., $\widetilde{\nabla}_{\diff\lambda(X)} \diff \lambda(Y) = \nabla_X Y$, for all $X,Y\in \Gamma(TM)$.

On the other hand, given a Galilean spacetime $\spacetime$ and a field of observers $Z\in \Gamma(TM)$, it is well-known that the Galilean connection $\nabla$ can be characterized using $Z$. The \textit{gravitational field} induced by $\nabla$ in  $Z$ is the spacelike vector field $\G =\nabla_Z Z$ and similarly, the \textit{vorticity} or \textit{Coriolis field} is the 2-form $\omega = \tfrac12 \rot(Z)$. The space of symmetric Galilean connections is mapped bijectively onto $\Gamma(\An) \times \Lambda^2(\An)$, giving $\nabla \mapsto (\G_Z, \omega_Z)$ (see \cite[Cor. 28]{Bernal}). 
Moreover, the converse is obtained through a formula 'à la Koszul'. If we consider
$$
\nabla_X Y=P^Z\left(\nabla_X Y\right)+X(\Omega(Y)) Z, \quad \forall X, Y \in \Gamma(T M)
$$
where $P^Z X=X-\Omega(X) Z$ is the natural spacelike projection for $Z$, then for each $V \in \Gamma(\operatorname{An}(\Omega))$ \cite[eq. (13)]{Bernal},
\begin{equation}\label{eq:Koszul}
\footnotesize
\begin{aligned}
2 g\left(P^Z\left(\nabla_X Y\right), V\right)= & X\left(g\left(P^Z Y, V\right)\right)+Y\left(g\left(P^Z X, V\right)\right)-V\left(g\left(P^Z X, P^Z Y\right)\right) \\
& +2 \Omega(X) \Omega(Y) g\left(\mathcal{G}^Z, V\right) +2 \Omega(X) \,\omega_{Z}(P^Z Y, V)+2 \Omega(Y)\, \omega_{Z}(P^Z X, V) \\
& +\Omega(X)\left(g\left(\big[Z, P^Z Y\big], V\right)-g\left([Z, V], P^Z Y\right)\right) \\
& -\Omega(Y)\left(g\left(\big[Z, P^Z X\big], V\right)+g\left([Z, V], P^Z X\right)\right) \\
& +g\left(\big[P^Z X, P^Z Y\big], V\right)-g\left(\big[P^Z Y, V\big], P^Z X\right) -g\left(\big[P^Z X, V\big], P^Z Y\right) .
\end{aligned}
\end{equation}

\subsection{Spatially conformally Leibnizian spacetimes}
\label{sescls}
Let us recall the notion of spatially conformally Leibnizian vector field, which appears in several classes of cosmological models in the context of the generalized Newton-Cartan theory (see \cite{GGRW}) and will be key to obtaining some of our results.
\begin{defi}
%\upshape
Let $(M,\Omega,g)$ be a Leibnizian spacetime and $K$ a vector field satisfying 
\begin{equation}\label{wdefi}
\Omega([K,V])=0\quad\mathrm{for}\;\mathrm{all}\quad V\in\Gamma(\mathrm{An}(\Omega)).
\end{equation}
 The vector field $K$ is called \emph{spatially conformally Leibnizian} if the Lie derivative of the  absolute space metric satisfies
\begin{equation}\label{conforme}
{\cal L}_{_K}g=2\rho\,g,
\end{equation}
for some smooth function $\rho\in C^{\infty}(M)$. The function $\rho$ is called \textit{conformal factor} of $g$ associated to $K$.
\end{defi} 

Notice that condition (\ref{wdefi}) ensures that the previous notion is well defined. As a direct consequence, we have
\begin{equation}
K(g(V,W))=2\rho\,g(V,W)+g([K,V],W)+g([K,W],V),\,
\end{equation}
for all $V,W\in\Gamma(\An)$.

The next result provides another condition to guarantee that \eqref{wdefi} holds for a vector field.
\begin{pro}\label{pro:spatially invariant}
Let  $(M,\Omega,g,\nabla)$ be a Galilean spacetime with symmetric connection $\nabla$. Then, a vector field $K$ satisfies equation {\rm (\ref{wdefi})} if and only if the function $\Omega(K)$ is spatially invariant, i.e., $V(\Omega(K))=0$, $\forall V\in\Gamma\big(\mathrm{An}(\Omega)\big)$.
\end{pro}

Finally, $K\in \Gamma(TM)$ is called a {\it spatially Leibnizian} vector field if the conformal factor $\rho$ vanishes identically, or equivalently, if its local flows $\Phi_s$ preserve the absolute space, i.e.,  $\Phi_s^* g = g$. If, in addition, they also preserve the absolute clock, i.e., $\Phi_s^* \Omega = \Omega$, then $K$ is called \textit{Leibnizian}. 

\subsection{Spacelike differential map and spacelike gradient}
\label{subsection:gradient}
To conclude the preliminaries section, let us introduce some notions that will be used later. Given a smooth function $\varphi\in C^\infty (M)$ on a Leibnizian spacetime $(M, \Omega, g)$, we  define the \textit{spacelike gradient} of $\varphi$, as the section $\grad \varphi\in \Gamma(\An)$, determined by
\[
    g(\grad\varphi, V) = \diff \varphi(V) , \quad \forall V\in \Gamma(\An)\,.
\]
Implicitly $\Omega(\grad \varphi) = 0$. This notion is a generalization of the gradient in semi-Riemannian manifolds.

When a smooth manifold $M$ splits as a product $M=I\times F$, for $I\subseteq \R$ an open interval and $F$ certain smooth manifold, the differential of a function $\varphi\in C^\infty(I \times F)$ can be canonically split. We can identify $T_{(t,p)} (I\times F) \equiv T_t I \times T_p F \equiv \R \times T_p F$. Then, define 
\[
\begin{split}
    \varphi_t\in C^\infty(I\times F), \quad{}& \varphi_t (t,p) \coloneq \diff \varphi_{(t,p)}(\pt)\,,\\
    \diff^F \varphi_{(t,p)}\colon T_pF\to \R, \quad{}& \diff^F \varphi_{(t,p)}(v) \coloneq \diff \varphi_{(t,p)}(\overline{v})\,,
\end{split}
\]
where $\pt$ is the canonical vector field associated to the projection onto the first component $t=\pi_I$ and $\overline{v}\in \R\times T_pF$ is the horizontal lift of a vector $v\in T_pF$. If $(F,h)$ is a Riemannian manifold, the usual Riemannian gradient $\grad^h$ applies. Thus, we will denote $\diff^h = \diff^F$ to unify the notation although the differential map only depends on the differential structure. From now on, the horizontal lift of a vector field $W\in \Gamma(TF)$ to a product $I\times F$ will be denoted by $\overline{W} \in \Gamma(T(I\times F)) \equiv \Gamma(\R \times TF)$. Moreover, the induced differential on $F$ can be lifted as 
\[
    \overline{\diff^h} \varphi_{(t,p)} \colon \R \times T_pF \to \R, \quad \overline{\diff^h} \varphi_{(t,p)}(a\pt + \overline{v}) \coloneq \diff^h \varphi_{(t,p)}(v)\,,
\]
for all $a\in \R$ and $v\in T_pF$.

On the other hand, we say that $\varphi\in C^\infty(I\times F)$ is bounded from below (resp. from above) \textit{along finite times} if, $\forall [s_1,s_2]\subset I$, there exists a constant $C_{s_1,s_2}\in \R$ such that $\varphi >C_{s_1,s_2}$ (resp. $\varphi <C_{s_1,s_2}$) in $[s_1,s_2]\times F$. A smooth function $\varphi$ is \textit{bounded along finite times} if it is bounded from below and from above along finite times. Obviously, $\inf \varphi>A$, $A\in \R$, along finite times is equivalent to $\varphi$ being bounded from below along finite times for constants greater than $A$.

\section{Galilean Twisted Spacetimes}\label{GT}
We introduce a new class of Galilean spacetimes, which extends the family of cosmological models known as Galilean Generalized Robertson-Walker spacetimes (see \cite{GGRW,Muler1983}).

\begin{defi}\label{defi:GT}
%\upshape
    Let $I\subseteq \R$, $0\in I$, be an open real interval, $(F,h)$ be an n-dimensional connected Riemannian manifold, and $f\in C^{\infty}(I\times F)$ be a positive smooth function. A Galilean spacetime $(M,\Omega, g, \nabla)$ is called \emph{Galilean Twisted} (GT) spacetime if $M=I\times F$, $\Omega = \diff \pi_I$, $g$ is the restriction to the vector bundle $\An$ of the following (degenerate) metric on M:
    \begin{equation}
        \overline{g} = f^2 \pi_F^* h\,,
    \end{equation}
    where $\pi_I, \pi_F$ are the canonical projections onto the open interval I and the manifold $F$, respectively, and $\nabla$ is the only symmetric Galilean connection on $M$ such that
    \begin{equation}\label{eq:GT conditions}
        \nabla_{\partial_t} \partial_t=0 \quad \text { and } \quad \operatorname{Rot} \partial_t=0\,,
    \end{equation}
    where $\partial_t=\partial / \partial t$ is the global coordinate vector field associated with $t\coloneq\pi_I$.
\end{defi}

Abusing the notation, we will denote $g=f^2 \pi_F^* h$.
Analogously to the Riemannian twisted product, the Riemannian manifold $(F,h)$ is called fiber.

The vector field $\pt$ is a field of observers in $M$, and the observers in $\partial_t$ are called \textit{commovil observers} in analogy with both the Galilean and relativistic Generalized Robertson-Walker spacetimes \cite{AliasRomeroSanchez1995, GGRW}. Notice that the conditions (\ref{eq:GT conditions}) determine the Galilean connection $\nabla$ (see equation (\ref{eq:Koszul})).

\begin{ex}
    Consider the GT spacetime $(\R\times \R^n, \diff \pi_1, g=f^2 \pi_2^* \langle\,, \rangle, \nabla)$, where $\pi_1 \colon \R\times \R^n \to \R$ and $\pi_2\colon \R\times \R^n \to \R^n$ are the canonical projections onto the first and second factors, respectively, $\langle\,,\rangle$ is the usual Euclidean product on $\R^n$ and $f\in C^\infty(\R\times \R^n)$. If $f$ does not depend on the fiber $F=\R^n$, i.e.,  $f = f(t)$, $t\in \R$, then we have a GGRW spacetime \cite{GGRW}. In the simplest case where $f=constant$, the Galilean connection is just the standard flat connection of $\R^{n+1}$.  
\end{ex}

\section{Geodesic Completeness of GT spacetimes}\label{sect4}

In this section we are going to study the completeness of the geodesics in a GT spacetime. In particular, this will allow us to know if the free falling observers live forever.

\begin{rem}
\emph{
    From the compatibility condition of a Galilean connection $\nabla$ with the absolute clock $\Omega$ in a Galilean spacetime, it is clear that the geodesics $\gamma$ maintain their causal character. That is, $\Omega(\gamma')$ is constant along the trajectory of $\gamma$. Then, the two relevant cases to consider will be spacelike geodesics ($\Omega(\gamma')=0$) and free falling observers ($\Omega(\gamma')=1$).
    }
\end{rem}

Given a geodesic $\gamma: J\longrightarrow M$ in a GT spacetime $(M=I\times F, \Omega, g, \nabla)$, where $J\subset\R$ denote an interval in the real line, we will be able to derive the corresponding equation of its projection  $\pi_F(\gamma)$ on the fiber $F$. The following lemma gives us the necessary computations.
\begin{lem}
    Let $(M=I\times F, \Omega= \diff \pi_I, g=f^2 \pi_F^* h, \nabla)$ be a GT spacetime. If $W_1,W_2\in \Gamma(TF)$, then
    \begin{equation}\label{eq:GT connection}
    \begin{split}
        \nabla_{\pt} \overline{W} ={}& \nabla_{\overline{W}} \pt = u_t \overline{W}\,,\\
        \nabla_{\overline{W_1}} \overline{W_2} ={}& \overline{\nabla^h_{W_1} W_2} + \diff^h u(W_1) \overline{W_2} + \diff^h u(W_2) \overline{W_1} - h(W_1,W_2) \overline{\grad^h u}\,,
    \end{split} 
    \end{equation}
    where $u\coloneq\log f$ and $\grad^h u$ denotes the usual gradient on the Riemannian manifold $(F,h)$.
\end{lem}
\begin{proof}
    It is enough to apply equation (\ref{eq:Koszul}) for $V=\overline{U}$, $U\in \Gamma(TF)$, since a local reference frame $\{E_i\}_{i=1}^n$ in $TF$ lift to a local reference frame $\{\pt\}\cup \{\overline{E_i}\}_{i=1}^n$ in $TM$.
\end{proof}

The previous formulae (\ref{eq:GT connection}) allow us to obtain the equations of the projections of the geodesics on the fiber of a GT spacetime. 
\begin{rem}
\upshape
    Let $\gamma$ be a geodesic $\gamma: J\longrightarrow M$ in a GT spacetime $(M=I\times F, \Omega, g, \nabla)$, where $J\subset\R$ denote an interval in the real line. If $\gamma$ is a free falling observer, we have
    \[
        \gamma(s) = (s, \sigma(s)) \quand \gamma'(s) = (\pt \circ \gamma) (s) + \overline{\sigma'}(s)\,,
    \]
    where $\sigma$ is a smooth curve in $F$. Consider $X$ a smooth vector field, which extends $\sigma'$ on a neighbourhood of $\sigma(s)\in F$. Then the lift $\overline{X}$ extends $\overline{\sigma'}$ and satisfies $L_{\pt} \overline{X} = [\pt, \overline{X}] = 0$. Thus, the geodesic equation turns out to be
    \begin{equation*}
    \begin{split}
        0 ={}& \nabla_{\gamma'} \gamma' = \nabla_{\gamma'} (\pt + \overline{X})|_\gamma \\
        ={}& \big(\nabla_{\pt} \pt + \nabla_{\overline{\sigma'}} \pt + \nabla_{\gamma'} \overline{X} \big)|_\gamma \\
        ={}& \big(2\nabla_{\pt} \overline{X} + \nabla_{\overline{\sigma'}} \overline{X}\big)|_\gamma\,,
    \end{split}
    \end{equation*}
    where we have $\nabla_{\pt} \pt=0$. Using (\ref{eq:GT connection}), 
    % we get
    % \begin{equation*}
    %     \nabla_{\overline{\sigma'}} X|\gamma = \overline{\tfrac{D^h \sigma'}{ds}} + 2 \diff^h u(\sigma') \overline{\sigma'} - h(\sigma', \sigma') \overline{\grad^h u}\,.
    % \end{equation*}
    the free falling observers $\gamma(s) = (s,\sigma(s))$ are characterized by the following equation on $(F,h)$,
    \begin{equation}\label{eq:completeness 1}
        0 = \mfrac{D^h \sigma'}{ds} + 2[u_t + \diff^h u(\sigma')] \sigma' - |\sigma'|^2_h \grad^h u\,.
    \end{equation}

    In the case $\gamma(s) = (c,\sigma(s))$, $c\in I$, is a spacelike geodesic, then $\gamma'=\overline{\sigma'}$. The corresponding equation on $(F,h)$ is given by
    \begin{equation}\label{eq:completeness spacelike}
        0 = \mfrac{D^h \sigma'}{ds} + 2\diff^h u(\sigma') \sigma' - |\sigma'|^2_h \grad^h u\,.
    \end{equation}
    Note that this is now an autonomous equation in the Riemannian manifold $(F,h)$. 
    \hfill $\diamond$
\end{rem}
For the rest of the section we suppose $I=\R$ since it is necessary to ensure the completeness of commovil observers.

Let $\spacetime$ be a GT spacetime with $M=\R\times F$ and $g=f^2 \, \pi_F^* h$. We can give a unique vector field $G$ on $T(\R\times F)$ such that, the timelike geodesic in $M$ are represented by integral curves of $G$ in $T(\R\times F)$. Indeed, let $\gamma(s)=(s,\sigma (s))$ be a free falling observer, $s_0\in J$ and $(U,x^1,...,x^n)$ a coordinate system in $F$, such that $\sigma(s_0)\in U$. If we identify $x(\sigma(s))\equiv x(s)=(x^1(s),...,x^n(s))$, equation (\ref{eq:completeness 1}) in $\R\times U$ is given by
\begin{multline}
\ddot{x}^k(s)+\sum_{i,j} \dot{x}^i(s)\dot{x}^j(s)\Gamma_{ij}^k(x(s))+2\big[u_t(s, x(s))
+\sum_l \mfrac{\partial u}{\partial x^l}(s, x(s))\dot{x}^l(s)\big]\dot{x}^k(s)\\
-\Big(\sum_{\alpha,\beta}h_{\alpha,\beta}(x(s))\dot{x}^\alpha(s)\dot{x}^\beta(s)\Big)\Big(\sum_\alpha h^{\alpha k}(x(s))\mfrac{\partial u}{\partial x^\alpha}(s, x(s))\Big)=0,
\end{multline}
where $k=1,...,n$ and the $\Gamma$'s denote the Christoffel symbols of the Riemannian metric $h$. Making use of the standard procedure, we can introduce auxiliary variables $v^i=\dot{x}^i$ to convert the previous second-order system, in the following equivalent first-order system in twice the number of variables,
\begin{equation}\label{eq:system}
\left\{
\begin{aligned}
\dot{x}^k(s)={}& v^k(s)\,,\\
\dot{v}^k(s)={}& -\sum_{i,j} v^i(s)v^j(s)\Gamma_{ij}^k(x(s))-2\big[u_t(t, x(s))
-\sum_l \mfrac{\partial u}{\partial x^l}(s, x(s))v^l(s)\big]v^k(s) \\
{}& + \Big(\sum_{\alpha,\beta}h_{\alpha,\beta}(x(s))v^\alpha(s)v^\beta(s)\Big)\Big(\sum_\alpha h^{\alpha k}(x(s))\mfrac{\partial u}{\partial x^\alpha}(s, x(s))\Big)\,,
\end{aligned}
\right.
\end{equation}
for $k=0,...,n$.

Treating $(x^1,...,x^n,v^1,...,v^n)$ as	coordinates on $T(\R\times U)$, we recognize  (\ref{eq:system}) as the equations for the flow of the vector field $G \in\mathfrak{X}(\R\times U)$ given by
\begin{multline}
G{(s,x,v)}=\sum_{k}\Big(v^k \mfrac{\partial}{\partial x^k}{\Big|_{(s,x,v)}}-\Big\lbrace\sum_{i,j} v^iv^j\Gamma_{ij}^k(x)+2\big[u_t(s, x(s)) + \sum_l \mfrac{\partial u}{\partial x^l}(s, x(s))v^l\big]v^k\\
-\Big(\sum_{\alpha,\beta}h_{\alpha,\beta}(x(s))v^\alpha v^\beta\Big)\Big(\sum_\alpha h^{\alpha k}(x(s))\mfrac{\partial u}{\partial x^\alpha}(s, x(s))\Big)\Big\rbrace \mfrac{\partial}{\partial v^k}{\Big|_{(t,x,v)}}\Big)\,.
\end{multline}

To show that $G$ define a global well-defined vector field on $T(\R\times F)$, it is enough to see that $G$ acts on a smooth function $\varphi$ by

$$G(\varphi)(p,v)=\tfrac{\diff}{\ds}{\big|_{s=0}} \varphi(\gamma_v(s),\dot{\gamma}_v(s))$$
\noindent where $\gamma_v(s)$ denote a geodesic satisfying $\gamma(0)=p$ and $\dot{\gamma}_v(0)=v$, $p\in \R\times F$ and $v\in T_p(\R\times F)$. This last statement is clear from the previous coordinates expressions.

Making use of \cite[Lem. 1.56]{ONeill1983}, we can assert that an integral curve $\xi$ of $G$ on an interval $[0,b[\subset \R$, for $b<+\infty$, can be extended to $b$ (as an integral curve) if and only if there exists a sequence in $[0,b[$, $\{s_n\}\nearrow b$, such that $\{\xi(s_n)\}$ converges.
Therefore, we obtain the following technical lemma. 
\begin{lem}\label{lem:extensibility}
    A free falling observer $\gamma\colon [0,b[ \to M$, $0<b<+\infty$, in a GT spacetime $\spacetime$ can be extended to $b$ as a geodesic of $\nabla$ if and only if there exists a sequence in $[0,b[$, $\{s_n\}\nearrow b$, such that $\{\gamma'(s_n)\}$ converges in $TM$.
\end{lem}

The following result gives natural conditions to ensure completeness of the free falling observers in a GT spacetime. 
\begin{theor}\label{theor:completeness free falling}
    Let $(F,h)$ be a complete Riemannian manifold and let $f\in C^\infty(\R \times F)$ be a positive function. If $\inf f >0$ and $\inf (f_t/f)>-\infty$ along finite times, then the trajectories of free falling observers in the GT spacetime $(M=\R\times F, \Omega =\diff \pi_\R, g=f^2\, \pi_F^* h, \nabla)$ are complete.
\end{theor}
\begin{proof}
    Given $0<b<+\infty$, consider a free falling observer 
    $$
    \gamma\colon [0,b[ \to \R \times F, \quad \gamma(s) = (s,\sigma(s))\,,
    $$
    where $\sigma\colon [0,b[ \to F$ is a smooth curve in $F$. The conditions on the function $f$ admit to applying the following argument for past extensibility analogously. Hence, if we show $\gamma$ is extensible towards $b$, the proof will be complete. 
    
    From equation (\ref{eq:completeness 1}) we have
    \begin{equation*}
        \tfrac{\diff}{\ds} h(\sigma', \sigma') = -2 \big[ 2u_t+\diff^h u(\sigma') \big] h(\sigma', \sigma')\,.
    \end{equation*}
    Note that 
    $$\tfrac{\diff}{\ds} (u\circ \gamma) = \diff u(\gamma') = \diff u\big(\pt + \overline{\sigma'}\big) = u_t + \diff^h u(\sigma')\,.$$
    Therefore, the square $h-$norm of $\sigma'$ is given by
    \begin{equation}\label{eq:completeness timelike hnorm}
    \begin{split}
        h(\sigma',\sigma') 
        ={}& A \exp\Big(-2\int_{0}^s u_t(\gamma(r)) \diff r\Big) \exp\big(-2 u(\gamma(s))\big)\\
        ={}& \tfrac{A}{f^2 (\gamma(s))} \exp\Big(-2\int_{0}^s u_t(\gamma(r)) \diff r\Big) \,,
    \end{split}        
    \end{equation}
    for some $A\in \R$ determined by the initial values $\gamma(0)$ and $\gamma'(0)$. 
    Hence, given $T>0$ such that $[0,b[\subset [0,T]$, the hypothesis on $f$ imply $h(\sigma',\sigma')$ is bounded. 
    Calling Lemma \ref{lem:extensibility}, we get $\gamma$ is extensible.
\end{proof}

\begin{rem}
\upshape
In our case (finite dimension), it is not necessary to endow the tangent fiber bundle with a Riemannian metric. Indeed, it is no difficult to see, under the assumptions of Lemma \ref{lem:extensibility}, that each integral curve of the vector field $G$ can be included in a compact subset of $TM$. 
\end{rem}

In the case of spacelike geodesics, we obtain analogous results.
\begin{theor}\label{theor:completeness spacelike}
    Let $(F,h)$ be a complete Riemannian manifold and let $f\in C^\infty(\R \times F)$ be a positive function. If $\inf f >0$ along finite times, then the spacelike geodesics in the GT spacetime $(M=\R\times F, \Omega =\diff \pi_\R, g=f^2\, \pi_F^* h, \nabla)$ are complete.
\end{theor}
\begin{proof}
    Given a spacelike geodesic $\gamma(s) = (c,\sigma(s))$, we can reason analogously by defining a suitable vector field on $TF$, whose integral curves are determined by equation (\ref{eq:completeness spacelike}). Now, 
    \[
        \tfrac{\diff}{\ds} h(\sigma', \sigma') = -2 \diff^h u(\sigma') h(\sigma', \sigma')\,.
    \]
    and $\dds (u\circ \gamma) = \diff^h u(\sigma')$. Therefore,
    \begin{equation}\label{eq:completeness spacelike hnorm}
        h(\sigma',\sigma') = \tfrac{A}{f^2 (\gamma(s))},\quad A\in \R\,.
    \end{equation}
    Using now $\inf f >0$ on $[0,T] \supset [0,b[$, we have $h(\sigma',\sigma')$ is bounded. An easy adaptation of Lemma \ref{lem:extensibility} imply $\gamma$ is extensible. 
\end{proof}

Note that, if the fiber $(F,h)$ is complete, Theorem \ref{theor:completeness spacelike} assures the geodesic completeness of any Riemannian manifold $(F, g_t)$, for $g_t = f^2(t,\cdot) \pi_F^* h$, under a natural assumption on $f$. 

We can ensure the geodesic completeness of a GT spacetime by unifying the results on the completeness of free falling observers and spacelike geodesics.
\begin{theor}\label{theor:completeness}
    Let $(F,h)$ be a complete Riemannian manifold and let $f\in C^\infty(\R \times F)$ be a positive function. If $\inf f >0$ and $\inf (f_t/f)>-\infty$ along finite times, then the GT spacetime $(M=\R\times F, \Omega =\diff \pi_\R, g=f^2\, \pi_F^* h, \nabla)$ is geodesically complete.
\end{theor}

If $F$ is compact, the hypothesis for $f\in C^\infty(\R \times F)$ in Theorem \ref{theor:completeness} are fulfilled trivially.
\begin{cor}\label{corollary:completeness compact}
    Let $(F,h)$ be a compact Riemannian manifold and $f\in C^\infty(\R \times F)$ a positive function. Then, the GT spacetime 
    $(M=\R\times F, \Omega =\diff \pi_\R, g=f^2\, \pi_F^* h, \nabla)$ is geodesically complete. 
\end{cor}

The following example prove that the conditions in Theorem \ref{theor:completeness} are indeed optimal.
\begin{ex}
\upshape
    Consider the Riemannian manifold $(F,h) = (\R, \cdot)$, to be the one-dimensional Euclidean space and $f\in C^\infty(\R^2)$ a smooth positive function. They define the GT spacetime 
    $$(\R\times \R, \Omega=\dt, g=f^2 \pi_2^*(\cdot) \equiv f^2 \diff x^2, \nabla)\,, $$ where $t\coloneq \pi_1$ and $x\coloneq \pi_2$.
    Then, free falling observers $\gamma(s) = (s,\sigma(s))$ satisfy the following second-order nonlinear differential equation (see (\ref{eq:completeness 1})), 
    \begin{equation*}
        0 = \diff \pi_2\big(\tfrac{D\gamma'}{ds}\big) = \sigma'' + 2\tfrac{f_t}{f}(s,\sigma(s))\,\sigma' + \tfrac{f_x}{f}(s,\sigma(s)) (\sigma')^2 \,.
    \end{equation*}
    From (\ref{eq:completeness timelike hnorm}) (which can be integrated explicitly in this case), we know 
    \begin{equation}
        \sigma'(s) = \mfrac{C}{f(s,\sigma(s))} \exp\Big(-\int_0^s u_t(r,\sigma(r)) \diff r\Big), \quad C\in \R\,,
    \end{equation}
    for a constant $C=\big\|\overline{\sigma'}\big\|_{\{0\}\times \R} = f(0,\sigma(0))\, \sigma'(0)$ (see equation (\ref{eq:velocidad relativa comoviles GT})).

    If we consider an incomplete geodesic $\gamma=(\cdot,\sigma(\cdot))\colon [0,b[\to \R\times \R$, for $b>0$, then $\lim_{s\tendsto b^-} |\sigma' (s)| = \infty$. For example, if $\inf f>0$ along finite times, then it is necessary $\lim_{s\tendsto b^-} u_t(s,\sigma(s)) = -\infty$. Thus,
    \[
        \lim_{s\tendsto b^-} \sigma(s) = \pm\infty \quand \lim_{x\tendsto\pm \infty}u_t(b,x) = -\infty\,.
    \]
    On the other hand, if $\inf u_t >-\infty$ along finite times, then $\lim_{s\tendsto b^-} f(s,\sigma(s)) = 0$. Hence, 
    \[
        \lim_{s\tendsto b^-} \sigma(s) = \pm\infty \quand \lim_{x\tendsto\pm \infty}f(b,x) = 0\,.
    \]
    Note that the conditions of boundedness from below on $f$ and $f_t/f$ arise naturally when there are incomplete geodesics. 

    In the case $f\colon F=\R \to \R$ is time-independent, free falling observers $\gamma(s) = (s,\sigma(s))$ can be integrated as
    \begin{equation}\label{eq:example completeness 2}
        \sigma(s) = G^{-1}(Cs+D), \quad C,D\in \R\,,
    \end{equation}
    where $G$ is a primitive function of $f$. Note that $G$ is a diffeomorphism on $\R$ onto its image since $G' = f >0$. 

    For instance, if $f(x) =  (1+x^2)^{-1}$, free falling observers satisfy $\gamma(s) = (s,\sigma(s)) = (s,\tan (Cs+D))$, where $C,D\in \R$ are constants fixed by $\sigma(0)$ and $\sigma'(0)$. These are incomplete observers. %, since $Cs + D$ tends to $\pm \pi/2$.

    % If $f(x) = \exp(x)$, then $\sigma(s) = \log(Cs + D)$. \\
    % If $f(x) = 1 + x^2$, then $\sigma(s) = z^{1/3} - (1/z)^{1/3}$, where $z=\frac{\sqrt{9(Cs+D)^2+4}-3(Cs+D)}{2}$ and $C,D\in \R$.
\end{ex}

\section{Physical interpretation of geodesics}\label{sect-phys}
Several interesting physical observations can be obtained from the expressions in the previous section.

\begin{rem}
\upshape
    From equation (\ref{eq:completeness timelike hnorm}), in the conditions of previous theorems, the free falling observers $\gamma$ satisfy
    \begin{equation}\label{eq:velocidad relativa comoviles GT}
        \|P^{\pt} \gamma' (s)\|_{\f_s} =\|P^{\pt} \gamma'(0)\|_{\f_0}\,\, e^{-\int_0^s u_t(\gamma)}\,,
    \end{equation}
    where $\| \cdot\|_{\f_s}$ is the norm induced by $g$ on the spacelike slice $\f_s = \{s\}\times F$, for all $s\in \R$. For a spacelike geodesic $\gamma=(c,\sigma(\cdot))$, equation (\ref{eq:completeness spacelike hnorm}) imply 
    \begin{equation}
        \|P^{\pt} \gamma'(s)\|_{\f_c} = \|P^{\pt} \gamma'(0)\|_{\f_c}\,.
    \end{equation}
    This last property is obvious from the fact that $\sigma = P^{\pt} \gamma'$ is a geodesic in the Riemannian manifold $(\f_c,g_c)$.
\end{rem}

Now, equation (\ref{eq:velocidad relativa comoviles GT}) allows us to heuristically deduce the behaviour of the complete timelike geodesics in a GT spacetime in relation to certain properties of the twister function.
Consider two simultaneous events $(s_0,p), (s_0,q) \in M$. Their spacelike distance (in the absolute space) is given by
\begin{equation}\label{eq:Galilean distance}
    \text{dist}_{\f_{s_0}}(({s_0},p),({s_0},q))=\inf_{\xi\in C_{F}([0,1],p,q)} \int_0^1 f\big({s_0},\xi(r)\big)\, \big|\xi'(r)\big|_h \diff r\,,
\end{equation}
where $C_F([0,1],p,q)$ is the set of admissible curves in $F$ from $p$ to $q$ (see, for instance, \cite{DoCarmo-Riemannian}). Note that the function  $\text{dist}_{\f_{s_0}}(({s_0},p),({s_0},q))$ determine the distance between the commovil observer $\theta_p({s_0})=({s_0},p)$ and $\theta_q({s_0})=({s_0},q)$ in the absolute space $\f_{s_0}$. Moreover, if $\beta(r) = (s_0, \xi(r))$, $r\in [0,1]$, is a minimizing geodesic in $\f_{s_0}$ from $p$ to $q$, then we have
\begin{equation}\label{eq:Galilean distance geodesic}
    \text{dist}_{\f_{s_0}}(({s_0},p),({s_0},q))=  \int_0^1 \|\beta'\| \diff r = \|\beta'\|\,,
\end{equation}
where $\| \cdot \|$ is the norm induced by $g$ in $\f_{s_0}$.

Obviously, if $f_t>0$ (resp. $f_t<0$) the absolute space is expanding (resp. contracting) over time. If for instance $f_t>0$, then the spacelike component of the velocity $\gamma'(s)$ of a free falling observer (no commovil) $\gamma$, relative to the corresponding instant commovil observer, decreases in norm. Taking this into account, if we assume that the twister function $f_t$ tends to $+\infty$ for $t\nearrow +\infty$, each free falling observer $\gamma$ evolves asymptotically to commovil observer directions (see fig. (a)). If otherwise $f_t<0$ and we assume $f_t$ tends to $-\infty$ for $t\nearrow +\infty$ and $\inf f>0$, then every free falling  observer $\gamma$ evolves asymptotically to some spacelike direction (see fig. (b)).

\begin{figure}[htbp]
    \centering
    \begin{subfigure}[b]{0.45\textwidth}
        \centering
        \tdplotsetmaincoords{70}{110}
        \begin{tikzpicture}[scale=2.5, tdplot_main_coords]
            % Draw horizontal plane
            \filldraw[fill=black!20!white, draw=green!70!black, opacity=0.5] (-1.25,-1,0) -- (1.25,-1,0) -- (1.25,1,0) -- (-1.25,1,0) -- cycle;
            \node[black] at (1.5,1.5,0) {$\mathcal{F}_{s_0} = \{s_0\} \times F$};
            
            % Draw vertical lines and label them
            \foreach \x/\q in {-0.5/{q}, 0.5/{p}} {
                \draw[dashed] (\x,0,-1) -- (\x,0,1);
                \node[right] at (\x,-0.3,1.2) {$\theta_{\q}(s)$};
            }
            
            % Draw points at the intersections
            \filldraw[red] (-0.5,0,0) circle (0.5pt) node[right] {$(s_0,q)$};
            \filldraw[red] (0.5,0,0) circle (0.5pt) node[left] {$(s_0,p)$};
            
            % Draw vectors
            \draw[->, thick, blue] (-0.5,0,0) -- ++(0,0,0.4) node[above right] {$\partial_t$};
            \draw[->, thick, blue] (0.5,0,0) -- ++(0,0,0.4) node[above left] {$\partial_t$};
            \draw[blue, thick, ->] (0.5,0,0) -- ++(-0.5,0.01,0) node[below right] {$\gamma'(s_0)$};
            
            % Draw curve
            \draw[black, thick, dashed] (1.3,0,-0.1) .. controls (0.75,0,-0.15) and (0.75,0,-.05) .. (0.5,0,0);
            \draw[black, thick] (0.5,0,0) .. controls (0.25,0,0) and (0,0,0) .. (-0.4,0,0.5) node[left] {$\gamma$};
        \end{tikzpicture}
        \caption{$f_t(s_0,p)>0$}
        \label{fig:sub1}
    \end{subfigure}
    \begin{subfigure}[b]{0.45\textwidth}
        \centering
        \tdplotsetmaincoords{70}{110}
        \begin{tikzpicture}[scale=2.5, tdplot_main_coords]
            % Draw horizontal plane
            \filldraw[fill=black!20!white, draw=green!70!black, opacity=0.5] (-1.25,-1,0) -- (1.25,-1,0) -- (1.25,1,0) -- (-1.25,1,0) -- cycle;
            \node[black] at (1.5,1.5,0) {$\mathcal{F}_{s_0} = \{s_0\} \times F$};
            
            % Draw vertical lines and label them
            \foreach \x/\q in {-0.5/{q}, 0.5/{p}} {
                \draw[dashed] (\x,0,-1) -- (\x,0,1);
                \node[right] at (\x,-0.3,1.2) {$\theta_{\q}(s)$};
            }
            
            % Draw points at the intersections
            \filldraw[red] (-0.5,0,0) circle (0.5pt) node[right] {$(s_0,q)$};
            \filldraw[red] (0.5,0,0) circle (0.5pt) node[left] {$(s_0,p)$};
            
            % Draw vectors
            \draw[->, thick, blue] (-0.5,0,0) -- ++(0,0,0.4) node[above right] {$\partial_t$};
            \draw[->, thick, blue] (0.5,0,0) -- ++(0,0,0.4) node[above left] {$\partial_t$};
            \draw[blue, thick, ->] (0.5,0,0) -- ++(-0.2,0,0.25) node[below right] {$\gamma'(s_0)$};
            
            % Draw curve
            \draw[black, thick, dashed] (1,0,-0.7) .. controls (0.75,0,-0.25) and (0.75,0,-0.2) .. (0.5,0,0);
            \draw[black, thick] (0.5,0,0) .. controls (0.25,0,0.25) and (0,0,0.4) .. (-0.4,0,0.5) node[left] {$\gamma$};
        \end{tikzpicture}
        \caption{$f_t(s_0,p)<0$}
        \label{fig:sub2}
    \end{subfigure}
    %\caption{Conjunto de figuras}
    \label{fig:fig}
\end{figure}

Now, it is natural to ask how one free falling observer measures the speed of another one. If two points in the fiber $p,q\in F$ are closely enough, free falling observers passing through one of them can measure the velocity of the others. If the leaves of the foliation induced by $\Omega$ are affine Euclidean spaces, the classical notion of relative velocity of an observer with respect to another is determined by the existence of a position vector between them. As a generalization, we can give the following, 
\begin{defi}
    Let $\gamma,\sigma$ be two observers in a (symmetric) Galilean spacetime and $\mathcal{F}_{s_{0}}$ the spacelike leaf in the absolute time $t=s_{0}\in\mathbb{R}$. If there exists a spacelike geodesic $\beta\colon [0,1] \to \f_{s_0}$ from $\gamma(s_{0})$ to $\sigma(s_{0})$ in $\f_{s_0}$, the \emph{relative spacelike velocity} of $\gamma$ with respect to $\sigma$ at $t=s_0$, $\textbf{V}\big(\gamma(s_0),\sigma(s_0)\big)$, is defined as the spacelike projection (with respect to $\sigma'$) of the parallel transport along $\beta$ of $\gamma'(s_0)$, i.e., 
    \[
        \textbf{V}\big(\gamma(s_0),\sigma(s_0)\big)\coloneq P^{\sigma'}\big(\mathcal{P}^\beta_{\gamma(s_0),\sigma(s_0)} \gamma'(s_0)\big)\,,
    \]
    where $\mathcal{P}^\beta$ denotes the parallel transport along $\beta$.
\end{defi}

Recall that in a Riemannian manifold, the existence and uniqueness of geodesics between two given points is locally ensured. Globally, the existence result holds in the complete case. Moreover, the no existence of conjugate points guarantees globally the uniqueness.

Let us deduce some interesting geometric properties of this definition in a GT spacetime. The fixed reference observer will be taken as a commovil observer.
Indeed, given a free falling observer $\gamma$ in a GT spacetime, denote $\gamma'(s_0) = \theta_p(s_0) + v$, for $v\in (\An)_{(s_0,p)}$. The parallel transport $\mathcal{P}^\beta_{(s_0,p),(s_0,q)} (\gamma'(s_0)) = V(1)$ is defined by a parallel vector field $V$ along $\beta$ with $V(0) = \gamma'(s_0)$. This can be split into 
$$V(r) = a(r) \pt|_{(s_0,\beta(r))} + U(r)\,,$$
where $a$ is a smooth function on $[0,1]$ and $U$ is a spacelike vector field along $\beta$. The initial value condition $V(0) = \gamma'(s_0)$ turns out
\[
    a(0) = 1, \quad U(0) = v\,.
\]
From the fact that $V$ is parallel and $\beta$ can be lifted from a curve in $F$, it is easy to obtain that $a'=0$ and $\mfrac{D U}{\diff r} = -u_t \beta'$. 
% Since $V$ is parallel, we have 
% \[
% \begin{split}
%     0 = \mfrac{D V}{\ds} ={}& a' \pt + a \nabla_{\beta'} \pt + \mfrac{D U}{\ds} \\
%     ={}& 
% \end{split}
% \]
Therefore, $a\equiv 1$ and $U$ satisfy the initial value problem 
\[
    \mfrac{D U}{\diff r} = -u_t \beta', \quad U(0) = v\,.
\]

The modulus of the relative spacelike velocity of $\gamma'(s_0)$ measured by $\theta_q(s_0)$, $\big\|\textbf{V}\big(\gamma'(s_0),\theta_q(s_0)\big)\big\| = \|U(1)\| = g(U(1), U(1))$, can be explicitly obtained in this case. It is enough to compute 
\[
    \tfrac{\diff}{\diff r} \|U\|^2 = -2 u_t g(U,\beta'),\quad \tfrac{\diff}{\diff r} g(U,\beta') = -u_t \|\beta'\|^2\,,
\]
using $\beta$ is a geodesic. Obviously, $\|\beta'\|$ is a constant. Taking into account the initial values, we can integrate both two differential equations, obtaining
\begin{equation}
    \big\|\textbf{V}\big(\gamma'(s_0),\theta_q(s_0)\big)\big\| = \Big\|\Big(\int_0^1 u_t(s_0,\beta(r)) \diff r \Big) w - v \Big\|_{\f_{s_0}}\,,
\end{equation}
where $v=P^{\pt} \gamma'(0)$ and the vector $w\coloneq \beta'(0)$. This expression, takes into account the geometric structure of the slice $\f_{s_0}$, as well as, the distance between the points $(s_0,p)$ and $(s_0,q)$ (see equations (\ref{eq:Galilean distance}) and (\ref{eq:Galilean distance geodesic})). 
In particular, if $\gamma$ is a commovil observer, the velocity of $\theta_p(s_0)$ measured by $\theta_q(s_0)$ is 
\begin{equation}\label{eq:Galilean relative velocity commovil}
\begin{split}
    \big\|\textbf{V}\big(\theta_p(s_0),\theta_q(s_0)\big)\big\| ={}& \Big|\int_0^1 u_t(s_0,\beta(r)) \diff r \Big| \big\| \beta'\big\| \\
    ={}& \Big|\int_0^1 u_t(s_0,\beta(r)) \diff r \Big| \, \text{dist}_{\f_{s_0}} \big( (s_0,p), (s_0,q) \big)\,.
\end{split}
\end{equation}
where $|\cdot|$ is the absolute value function.
We see that it is proportional to the spacelike distance between the points $(s_0,p)$ and $(s_0,q)$.

Finally, if the GT spacetime can be reduced to a GGRW spacetime, $f=f(t)$, then expression (\ref{eq:Galilean relative velocity commovil}) turns out $ \|\textbf{V} \big({\theta_p(s_0)},\theta_q(s_0)\big)\| = u_t(s_0) \, \text{dist}_{\f_{s_0}} \big( (s_0,p), (s_0,q) \big)$.

\section{Galilean Torqued Vector Fields}\label{sect5}
Along this section, we will see that there exists a natural generalization of GT spacetimes in a rich geometrical structure characterized by the existence of a  distinguished  vector field on the Galilean spacetime. In fact, the main theorem will state that these Galilean spacetimes are locally GT spacetimes. 
\begin{pro}\label{pro:torqued vector field in GT}
    Let $(M=I\times F, \Omega = \dt, g = f^2 \pi_F^* h, \nabla)$ be a GT spacetime. Then, the vector field $K=f\pt$ satisfies
    \begin{equation}\label{eq:GT torqued}
        \nabla_X K = f_t X + \alpha(X) K \quand \alpha(K) = 0, \quad \forall X \in \Gamma(T M)\,,
    \end{equation}
    where $\alpha = \overline{\diff^h} (\log f) \in \bigwedge^1 (M)$ (see subsection \ref{subsection:gradient}).
\end{pro}
\begin{proof}
    Let $X\in \Gamma(TM)$ be a vector field. Consider $\{E_i\}_{i=1}^n$ a local reference frame in $TF$ that lifts to a local reference frame $\{\pt\}\cup \{\overline{E_i}\}_{i=1}^n$ in $TM$. Then, there exist local smooth functions $g_0,g_1,\ldots,g_n$ such that
    \begin{equation*}
        X=g_0 \pt + \sum_i g_i \overline{E_i}\,.
    \end{equation*}
    Then, if $u\coloneq \log f$, we have
    \begin{equation*}
    \begin{split}
        \nabla_X K ={}& X(f) \pt + f \nabla_X \pt \\
        %={}& X(u) K + f \sum_i g_i \nabla_{\overline{E_i}} \pt \\
        ={}& X(u) K + f \sum_i g_i u_t \overline{E_i} \\
        ={}& X(u) K + f u_t (X-g_0\pt) \\
        ={}& f_t X + \overline{\diff^h} u(X) K\,.
    \end{split}
    \end{equation*}
\end{proof}

As we announced, a generalization of GT spacetimes is defined by means of a vector field satisfying an analogous expression to equation (\ref{eq:GT torqued}).
\begin{defi}\label{defi:torqued}
%\upshape
    Let $(M, \Omega, g, \nabla)$ a Galilean spacetime. A future-pointing timelike vector field $K \in \Gamma(TM)$ satisfying
    \begin{equation}\label{eq:torqued}
        \nabla_X K = \rho X + \alpha(X) K \quand \alpha(K) = 0, \quad \forall X \in \Gamma(T M)\,,
    \end{equation}
    for some $\rho\in C^\infty(M)$ and $\alpha \in  \bigwedge^1 (M)$, is called a \textup{Galilean torqued vector field}. The function $\rho$ is called the \textup{torqued function} and $\alpha$ the \textup{torqued form}.
\end{defi}

The torqued function and form are fully determined by the function $\Omega(K)$, as we prove in the following lemma.
\begin{lem}
    Let $\spacetime$ be a spacetime and suppose it admits a Galilean torqued vector field $K\in \Gamma(TM)$ with torqued function $\rho\in C^\infty(M)$ and torqued form $\alpha \in \bigwedge^1(M)$. Then, 
    \begin{equation}\label{eq:alpha}
        \alpha(V) = \tfrac{1}{\Omega(K)} V\big(\Omega(K)\big), \quad\forall V\in \Gamma(\An)\,.
    \end{equation}
 and
    \begin{equation}\label{eq:rho}
        \rho = \tfrac{1}{\Omega(K)} K\big( \Omega(K)\big)\,.
    \end{equation}
\end{lem}
\begin{proof}
    Consider the field of observers $Z=\frac{1}{\Omega(K)} K$. Then, for any $V\in \Gamma(\An)$,
    \begin{equation}\label{eq:proof derivative Z}
    \begin{split}
        \nabla_V Z ={}& V\big(\tfrac{1}{\Omega(K)}\big) K + \tfrac{1}{\Omega(K)} \nabla_V K \\
        ={}& \tfrac{\rho}{\Omega(K)} V + \left[ \tfrac{\alpha(V)}{\Omega(K)} + V\big(\tfrac{1}{\Omega(K)}\big) \right] K\,.
    \end{split}
    \end{equation}
    Now, making use of the compatibility of the connection with $\Omega$ we have $\Omega(\nabla_V Z) = V(\Omega(Z)) = 0$. Then, from (\ref{eq:proof derivative Z}) we get (\ref{eq:alpha}).

    On the other hand, from equation (\ref{eq:torqued}) we have $\nabla_K K = \rho K$. Consequently,
    \[
        K(\Omega(K)) = \Omega(\nabla_K K) = \rho \Omega(K)\,.
    \]
    Thus, it follows expression (\ref{eq:rho}).
\end{proof}

\begin{rem}\label{rem:torqued}
\upshape
    Let $K$ be a Galilean torqued vector field in a Galilean spacetime $\spacetime$ with torqued function $\rho\in C^\infty(M)$ and torqued form $\alpha \in \bigwedge^1(M)$. Consider the associated field of observers $Z\coloneq\tfrac{1}{\Omega(K)} K$. Note that we have obtained in the proof of the lemma the following identity,
    \begin{equation}\label{eq:connection on Z}
        \nabla_V Z = \tfrac{\rho}{\Omega(K)} V, \quad \forall V\in \Gamma(\An)\,.
    \end{equation}
    Since $V(\Omega(Z)) = 0$ for all $V\in \Gamma(\An)$, from equation (\ref{eq:connection on Z}) we can compute 
    \begin{equation*}
    \begin{split}
        L_Z g(V,W) ={}& Z g(V,W) - g([Z,V],W) - g(V, [Z,W]) \\
        ={}& g(\nabla_V Z, W) + g(V, \nabla_W Z) \\
        ={}& 2 \tfrac{\rho}{\Omega(K)} g(V,W)\,,
    \end{split}
    \end{equation*}
    for all $V,W\in \Gamma(\An)$ (see subsection \ref{sescls}).
    Therefore, $Z$ is a spatially conformally Leibnizian vector field.

    \hfill $\diamond$
\end{rem}

Finally, we obtain the main theorem of this section, which determines the local structure of Galilean spacetimes admitting a Galilean torqued vector field.
\begin{theor}\label{theor:local structure torqued}
    Let $\spacetime$ be a Galilean spacetime. %(with symmetric connection $\nabla$). 
    If it admits a torqued vector field $K\in \Gamma(TM)$, then for each $p\in M$, there exists an open neighborhood of p, $\U$, and a Galilean diffeomorphism $\Psi\colon N \to \U$, where $N$ is a GT spacetime. 
\end{theor}
\begin{proof} 
    This is in fact a generalization of the corresponding proof of the local structure Theorem for ICL spactimes \cite[Th. 12]{GGRW}. In order to
make the paper self-contained we give the details.
    
    Suppose there exists a torqued vector field $K$. Let $\rho\in C^\infty(M)$ and $\alpha \in \bigwedge^1(M)$ be the associated torqued function and one-form, and $Z\coloneq \frac{1}{\Omega(K)} K$ be the associated field of observers. Given an event $p\in M$, take a neighbourhood $U_p$ of $p$ in the leaf $\mathcal{F}_p$ of the foliation induced by $\Omega$, and $I\subset \R$, $0\in I$, a suitable interval such that the local flow of $Z$,
    \[
    \begin{split}
        \Psi\colon I \times U_p \to{}& M\\
        (s,q) \longmapsto{}& \phi_s(q)\,,
    \end{split}
    \]
    is well-defined and one-to-one. It satisfies
    \[
        \diff \Psi_{(s,q)}(1,0) = Z_{\phi_s(q)}, \quad \diff \Psi_{(s,q)}(0,v) = \big(\diff \phi_s\big)_q(v)\,,
    \]
    for all $(s,q) \in I\times U_p$ and $v\in T_pU_p$. Now, if we identify $t\equiv \pi_I$ and $\pt \equiv (1,0)$, it follows by construction
    \[
        \Omega(\diff \Psi(\pt)) = \Omega (Z) = 1 = \dt (\pt)\,.
    \]
    Moreover, the flow $\phi_s$ of $Z$ preserves $\Omega$ on spacelike vector fields, $L_Z \Omega(V) = 0$, $\forall V\in \Gamma(\An)$. Then,
    $$
        \Omega\left(\diff \Psi_{(s, q)}(0, v)\right)=\Omega\left((\diff \phi_s)_q(v)\right)=(\phi_s^* \Omega)_q(v)=\Omega_q(v)=0 .
    $$
    Consequently, $\Psi^* \Omega=\dt$, and each level set of $t$ corresponds with certain open set $\mathcal{U} \cap \mathcal{F}_a$ of some leaf of the foliation $\mathcal{F}$ induced by $\Omega$.
    
    From Remark \ref{rem:torqued}, $Z$ is spatially conformally Leibnizian:
    \begin{equation}\label{eq:spatially conf leibnizian}
        L_Z g(V, W)=\tfrac{2 \rho}{\Omega(K)} g(V, W), \quad \forall V, W \in \Gamma(\operatorname{An}(\Omega)) \,.
    \end{equation}
    Consider $v,w\in T_q U_p$. If we define
    $$
    \eta_q(s)\coloneq \Psi^* g_{(s, q)}((0, v),(0, w))=g_{\phi_s(q)}\left(\big(\diff \phi_s\big)_q(v), \big(\diff \phi_s\big)_q(w)\right),
    $$
    it follows from equation (\ref{eq:spatially conf leibnizian}),
    $$
    \eta_q^{\prime}(s)=\left.\tfrac{2 \rho}{\Omega(K)}\right|_{\phi_s(q)} \eta_q(s)\,.
    $$
    Then, taking
    \begin{equation}\label{eq:scale factor}
        f(s,q)\coloneq\exp \left(\int_0^s \frac{\rho}{\Omega(K)}\left(\phi_l(q)\right) \diff l\right)\,,
    \end{equation}
    we obtain $\eta_q(s)=f^2(s,q) g_q(v,w)$. Therefore, $\Psi^* g(s,q)=f^2(s,q) g_q$.
    
    Finally, from equation (\ref{eq:connection on Z}), and any arbitrary tensor field $V\in \Gamma(\An)$, we have
    $$
    \operatorname{Rot}(Z)(V, W)=g\left(\nabla_V Z, W\right)-g\left(\nabla_W Z, V\right) =0\,,
    $$
    and 
    \begin{equation*}
        \nabla_Z Z = \tfrac{\rho}{\Omega(K)} Z + Z\big(\tfrac{1}{\Omega(K)}\big) K = 0\,.
    \end{equation*}
    Thus, $Z$ is an irrotational field of observers. Since $\diff\Psi(\pt) = Z$, from \cite[Cor. 28]{Bernal}, the Galilean connection $\nabla$ is $\Psi-$related to the corresponding GT connection $\nabla^{GT}$ in $I\times U_p$. Therefore, $\Psi$ is a Galilean diffeomorphism on 
    \begin{equation*}
        \big(I\times U_p, \diff \pi_I, f^2\, \pi_{U_p}^* g, \nabla^{GT}\big)\,.
    \end{equation*}
\end{proof}

Note that the conclusions are similar when we consider ICL spacetimes \cite{GGRW}. In this case, the distinction lies in the fact that the torqued factor is not necessarily spatially uniform.

\section{Global splitting}\label{sect6}
This section is devoted to give natural geometric conditions to characterize a Galilean spacetime as a global GT spacetime. We will use strongly the geometric properties satisfied by a Galilean torqued vector field.

In the following Lemma, we define an auxiliary  semi-Riemannian metric in terms of the Galilean structure of a Galilean spacetime. The spacelike gradient defined in subsection \ref{subsection:gradient} allows us to write the associated Levi-Civita connection in terms of the Galilean connection.
\begin{lem}\label{lem:Koszul}
    Let $\spacetime$ be a Galilean spacetime %with symmetric connection $\nabla$ 
    and let $Z\in \Gamma(TM)$ be an irrotational spatially conformally Leibnizian field of observers. Given a smooth function $\varphi\in C^\infty(M)$ and $\epsilon\in \{-1,+1\}$, define on $M$ the semi-Riemannian metric
    \[
    \overline{g}_{\epsilon,\varphi} \coloneq \epsilon \Omega \otimes \Omega + e^{2\varphi} g(P^Z\cdot, P^Z \cdot)\,.
    \]
    Then, the associated Levi-Civita connection $\overline{\nabla}$ of $\overline{g}_{\epsilon,\varphi}$ is given by 
    \begin{equation}\label{eq:semi-Riemannian connection}
    \begin{aligned}
        \overline{\nabla}_X Y= & \nabla_X Y+\diff \varphi(X) P^Z Y+\diff \varphi(Y) P^Z X-\Omega(X) \Omega(Y) \mathcal{G}^Z \\
        & -\overline{g}_{\epsilon,\varphi}(P^Z X, P^Z Y) \grad \varphi - \epsilon\big(Z(\varphi)+\Tilde{\rho}\big) \overline{g}_{\epsilon, \varphi}\left(P^Z X, P^Z Y\right) Z\,,
    \end{aligned}
    \end{equation}
    for all $X,Y\in \Gamma(TM)$, where $\mathcal{G}^Z$ is the gravitational field of $Z$, $\Tilde{\rho}$ is the conformal factor associated to $Z$ and $\grad \varphi$ is the spacelike gradient of $\varphi$ in the Galilean structure. 
\end{lem}
\begin{proof} 
    Consider $X,Y\in \Gamma(TM)$ and $V\in \Gamma(\An)$. Koszul formula \cite[Th. 11]{ONeill1983} gives
    \begin{equation}\label{eq:lem Koszul}
    \footnotesize
    \begin{aligned}
    2 \overline{g}_{\epsilon, \varphi}\left(\overline{\nabla}_X Y, V\right)= & X\left(\overline{g}_{\epsilon, \varphi}(Y, V)\right)+Y\left(\overline{g}_{\epsilon, \varphi}(X, V)\right)-V\left(\overline{g}_{\epsilon, \varphi}(X, Y)\right) \\
    & +\overline{g}_{\epsilon, \varphi}([X, Y], V)-\overline{g}_{\epsilon, \varphi}([X, V], Y)-\overline{g}_{\epsilon, \varphi}([Y, V], X) \\
    = & e^{2 \varphi}\left\{2 X(\varphi) g\left(P^Z Y, V\right)+X\left(g\left(P^Z Y, V\right)\right)+2 Y(\varphi) g\left(P^Z X, V\right)\right. \\
    & +Y\left(g\left(P^Z X, V\right)\right)-2 V(\varphi) g\left(P^Z X, P^Z Y\right)-V\left(g\left(P^Z X, P^Z Y\right)\right) \\
    & +g\left(\big[P^Z X, P^Z Y\big], V\right)-\Omega(X) g\left(\big[P^Z Y, Z\big], V\right)+\Omega(Y) g\left(\big[P^Z X, Z\big], V\right) \\
    & -g\left(\big[P^Z X, V\big], P^Z Y\right)+\Omega(X) g\left([V, Z], P^Z Y\right) \\
    & \left.-g\left(\big[P^Z Y, V\big], P^Z X\right)+\Omega(Y) g\left([V, Z], P^Z X\right)\right\}
    \end{aligned}
    \end{equation}
    Considering $\omega(Z)=0$, we can derive from equation (\ref{eq:Koszul}) the following expression:
    \begin{equation}\label{eq:lem Koszul Gal}
    \footnotesize
    \begin{aligned}
    2 \overline{g}_{\epsilon, \varphi}\left(\nabla_X Y, V\right) ={}& 2 e^{2 \varphi} g\left(P^Z\left(\nabla_X Y\right), V\right) \\
    ={}& e^{2 \varphi}\left\{X\left(g\left(P^Z(Y), V\right)\right)+Y\left(g\left(P^Z(X), V\right)\right)-V\left(g\left(P^Z(X), P^Z(Y)\right)\right.\right. \\
    & +2 \Omega(X) \Omega(Y) g\left(\mathcal{G}^Z, V\right) +\Omega(X)\left(g\left(\big[Z, P^Z(Y)\big], V\right)-g\left([Z, V], P^Z(Y)\right)\right) \\
    & -\Omega(Y)\left(g\left(\big[Z, P^Z(X)\big], V\right)+g\left([Z, V], P^Z(X)\right)\right)  +g\left(\big[P^Z(X), P^Z(Y)\big], V\right)\\ 
    & -g\left(\big[P^Z(Y), V\big], P^Z(X)\right) \left. -g\left(\big[P^Z(X), V\big], P^Z(Y)\right)\right\} \,.
    \end{aligned}
    \end{equation}
    Substituting (\ref{eq:lem Koszul Gal}) in (\ref{eq:lem Koszul}), we obtain 
    \[
    %\small
    \begin{aligned}
    \overline{g}_{\epsilon, \varphi}\left(\overline{\nabla}_X Y, V\right)={}& \overline{g}_{\epsilon, \varphi}\left(\nabla_X Y, V\right)+\diff \varphi(X) \overline{g}_{\epsilon, \varphi}(Y, V)+ \diff \varphi(Y) \overline{g}_{\epsilon, \varphi}(X, V) \\
    & -e^{2\varphi} V(\varphi) g\left(P^Z X, P^Z Y\right)-\Omega(X) \Omega(Y) \overline{g}_{\epsilon, \varphi}\left(\mathcal{G}^Z, V\right)\,.
    \end{aligned}
    \]

    On the other hand, applying (\ref{eq:Koszul}) again and the relation $L_Z g = 2 \Tilde{\rho} g$, we have
    \begin{equation}\label{eq:lem Koszul Z}
    \footnotesize
    \begin{aligned}
    2 \overline{g}_{\epsilon, \varphi}\left(\overline{\nabla}_X Y, Z\right) ={}& X\left(\overline{g}_{\epsilon, \varphi}(Y, Z)\right)+Y\left(\overline{g}_{\epsilon, \varphi}(X, Z)\right)-Z\left(\overline{g}_{\epsilon, \varphi}(X, Y)\right) \\
    & +\overline{g}_{\epsilon, \varphi}([X, Y], Z)-\overline{g}_{\epsilon, \varphi}([X, Z], Y)-\overline{g}_{\epsilon, \varphi}([Y, Z], X) \\
    ={}& \epsilon X(\Omega(Y))+\epsilon Y(\Omega(X))-\epsilon Z(\Omega(X)) \Omega(Y)-\epsilon \Omega(X) Z(\Omega(Y)) \\
    & -2 e^{2 \varphi} Z(\varphi) g\left(P^Z X, P^Z Y\right)-e^{2 \varphi} Z\left(g\left(P^Z X, P^Z Y\right)\right)+\epsilon \Omega([X, Y]) \\
    & -\epsilon \Omega([X, Z]) \Omega(Y)-e^{2 \varphi} g\left(P^Z[X, Z], P^Z Y\right) \\
    & -\epsilon \Omega([Y, Z]) \Omega(X)-e^{2 \varphi} g\left(P^Z[Y, Z], P^Z X\right) \\
    ={}& 2 \epsilon X(\Omega(Y))-2 e^{2 \varphi} Z(\varphi) g\left(P^Z X, P^Z Y\right) - 2 e^{2 \varphi} \Tilde{\rho} g\left(P^Z X, P^Z Y\right) .
    \end{aligned}
    \end{equation}
    Using $\overline{g}_{\epsilon, \varphi}\left(\nabla_X Y, Z\right)=\epsilon \Omega\left(\nabla_X Y\right)=\epsilon X(\Omega(Y))$, equation (\ref{eq:lem Koszul Z}) gives
    $$
    \begin{aligned}
    \overline{g}_{\epsilon, \varphi}\left(\overline{\nabla}_X Y, Z\right) ={}& \overline{g}_{\epsilon, \varphi}\left(\nabla_X Y, Z\right)-e^{2 \varphi} Z(\varphi) g\left(P^Z X, P^Z Y\right) - e^{2 \varphi} \Tilde{\rho} g\left(P^Z X, P^Z Y\right) \,.
    \end{aligned}
    $$
    
\end{proof}

We will describe how the semi-Riemannian induced structure on a Galilean spacetime can be split into a semi-Riemannian twisted product $I\times_f \f$, for an interval $I$, a Riemannian manifold $\f$ and certain positive smooth function $f\in C^\infty(I\times \f)$. This requires a slight adaptation of \cite[Th. 1]{Ponge1993}, in which the geodesic completeness of the leaves $L$ is replaced in the following special case. 

Consider a smooth $r$-dimensional manifold $N$ and suppose $N$ admits two smooth transverse foliations $\mathcal{L}$ and $\mathcal{K}$, $1$-dimensional the first one and $(r-1)$-dimensional the second one. We can assume that $\mathcal{L}$ is defined by a smooth vector field $X_{\mathcal{L}}$.

\begin{defi} 
The foliation $\mathcal{L}$ admits a \emph{uniform global generator} $X_{\mathcal{L}}$ if there exists a transverse leaf $\mathcal{K}_0\in \mathcal{K}$ and an open real interval $I$, such that the maximal flow $\phi$ of $X_{\mathcal{L}}$ is well-defined and onto on $I\times \mathcal{K}_0$. 
%satisfies that $\phi|_{I\times \mathcal{K}_0}: I\times \mathcal{K}_0\longrightarrow N$ is well defined and onto.
\end{defi}

The proof of \cite[Th. 1]{Ponge1993} remains valid in the case of the following lemma.
\begin{lem}\label{lem:Ponge modified}
    Let $(M,g)$ be a simply connected semi-Riemannian manifold endowed with two orthogonal foliations $\mathcal{L}$ and $\mathcal{K}$. Suppose $\mathcal{L}$ is one-dimensional and admits a uniform global generator $X\in \Gamma(TM)$ with associated interval $I\subset \R$ and transverse leaf $\mathcal{K}_0\in \mathcal{K}$. If $\mathcal{L}$ is totally geodesic and $\mathcal{K}$ is totally umbilic, then $(M,g)$ is isometric to a (semi-Riemannian) twisted product $I \times_f \mathcal{K}_0$. 
\end{lem}

The above lemma allows us to prove
\begin{theor}\label{theor:semi-Riemannian isometry}
    Let $\spacetime$ be a simply connected Galilean spacetime %with symmetric connection $\nabla$, 
    and let $K\in \Gamma(TM)$ be a Galilean torqued vector field. Given a smooth function $\varphi\in C^\infty(M)$ and $\epsilon\in \{-1,+1\}$, if the field of observers $Z\coloneq K/\Omega(K)$ is a uniform global generator with associated interval $I\subseteq \R$ and transverse leaf $\f$ in the foliation induced by $\Omega$, then
    the semi-Riemannian manifold
    \[
    (M,\overline{g}_{\epsilon,\varphi}), \quad \text{with} \quad \overline{g}_{\epsilon,\varphi} = \epsilon \Omega \otimes \Omega + e^{2\varphi} g(P^Z\cdot, P^Z \cdot)\,,
    \]
    is isometric to the (semi-Riemannian) twisted product
    \[
    (I \times \f, \overline{g}), \quad \text{with} \quad \overline{g} = \epsilon \, \pi_I^* \dt^2 +  f^2 \pi_{\f}^* \big(g|_{\f}\big)\,,
    \]
    where $g|_{\f}$ is the restriction of $g$ to $\f$ and $f$ is a positive smooth function on $I\times \f$. 

    Moreover, $Z$ is a spatially conformal vector field of $\overline{g}_{\epsilon,\varphi}$. %, i.e., it is a conformal vector field of the semi-Riemannian metric $\overline{g}_{\epsilon,\varphi}$ restricted to the leaves of the distribution $\An$.
    
    % If $\varphi$ is spatially invariant, then $(M,\overline{g}_{\epsilon,\varphi})$ is isometric to the warped product $(\R \times S, \overline{g})$.
    % If $K(\varphi) = - \rho$, where $\rho$ is the torqued function associated to $K$, then $(M,\overline{g}_{\epsilon,\varphi})$ is isometric to the product $(\R \times S, \overline{g})$, with $\overline{g} = \epsilon \dt^2 + g$.
\end{theor}
\begin{proof}
    Note first that we are in the hypothesis of Lemma \ref{lem:Koszul} with a conformal factor $\Tilde{\rho}=\rho/\Omega(K)$ (cf. equation (\ref{eq:connection on Z})).

    The induced foliation by $Z$ is parameterized by its integral curves $\phi_p(t)$, for $t\in I$ and $p\in \f$. From equation (\ref{eq:semi-Riemannian connection}) and using $\nabla_Z Z=0$, we obtain $\overline{\nabla}_Z Z = 0$. Therefore, the integral curves $\phi_p(t)$ are geodesics for both $\nabla$ and $\overline{\nabla}$, that is, the associated foliation is totally geodesic.
    
    Moreover, for all vector fields $V,W\in \Gamma(\An)$, we have
    \begin{equation*}
    \begin{split}
        2 \overline{g}_{\epsilon,\varphi}(\overline{\nabla}_V W, Z) ={}& -2\overline{g}_{\epsilon,\varphi}(W,\overline{\nabla}_V Z)\\
        ={}& - Z \overline{g}_{\epsilon,\varphi} (V,W) - \overline{g}_{\epsilon,\varphi}(V,[W,Z]) + \overline{g}_{\epsilon,\varphi}(W,[Z,V]) \\
        ={}& e^{2\varphi}\big[ -2 Z(\varphi) g(V,W) - g(V,\nabla_W Z) - g(\nabla_V Z, W) \big] \\
        ={}& -2\big( Z(\varphi) + \tfrac{\rho}{\Omega(K)}\big) e^{2\varphi} g(V,W)\\
        ={}& -2\big( Z(\varphi) + \tfrac{\rho}{\Omega(K)}\big) \overline{g}_{\epsilon,\varphi}(V,W)\,.
    \end{split}   
    \end{equation*}
    Then, $\overline{II}(V,W) = -\tfrac{K(\varphi) + \rho}{\Omega(K)}\, \overline{g}_{\epsilon,\varphi}(V,W) Z$; hence, the leaves of the foliation induced by $\An$ are totally umbilic for the semi-Riemannian metric $\overline{g}_{\epsilon,\varphi}$. From Lemma \ref{lem:Ponge modified}, $(M,\overline{g}_{\epsilon, \varphi})$ is isometric to a (semi-Riemmanian) twisted product $I\times_f \f$.

    Furthermore, from equation (\ref{eq:semi-Riemannian connection}) we have 
    \[
        \overline{\nabla}_V Z = \tfrac{1}{\Omega(K)} (\rho + K(\varphi)) V\,,
    \]
    for all spacelike vector fields $V$.
    Then,
    \[
    \begin{aligned}
        L_Z \overline{g}_{\epsilon,\varphi} (V,W) ={}& \overline{g}_{\epsilon,\varphi}(\overline{\nabla}_V Z,W) + \overline{g}_{\epsilon,\varphi}(V, \overline{\nabla}_W Z)\\
        ={}& \tfrac{2}{\Omega(K)} (\rho + K(\varphi)) \overline{g}_{\epsilon,\varphi} (V,W)\,,
    \end{aligned}
    \]
    for all spacelike vector fields $V,W$, that is, $Z$ is spatially conformal for $\overline{g}_{\epsilon,\varphi}$.

    % The last two statements are straightforward using the corresponding expression in (\ref{eq:semi-Riemannian connection}).
\end{proof}

Finally, we are in the position to obtain a converse for Proposition \ref{pro:torqued vector field in GT}. 
\begin{theor}\label{theor:splitting global}
    Let $\spacetime$ be a simply connected Galilean spacetime. % with symmetric connection $\nabla$.
    It admits a global decomposition as a GT spacetime if and only if there exists a Galilean torqued vector field $K\in \Gamma(TM)$ such that the associated field of observers, $Z\coloneq K/\Omega(K)$, is a uniform global generator relative to the foliation defined by $\An$. %is complete.
\end{theor}
\begin{proof}
    From Proposition \ref{pro:torqued vector field in GT}, the necessary condition is obvious.
    Now, if the flow of $Z$ is well-defined and onto on $I\times \f$, for $I\subset \R$ an open interval and $\f$ a leaf of the foliation induced by $\Omega$, Theorem \ref{theor:semi-Riemannian isometry} states that there exists a semi-Riemannian isometry map 
    $$
    \Psi\colon \big(M,\overline{g}_{-1,0}\big) \To \big(I\times \f, - \dt^2 + f^2 \, g|_{\f}\big)\,,
    $$
    where we take $\epsilon=-1$ and $\varphi=0$. 
    Consider $\Psi$ as a map between Galilean spacetimes,
    $$
    \Psi\colon \big(M,\Omega,g,\nabla\big) \To \big(I\times \f, \dt, f^2\, g|_{\f},\nabla^0\big)\,,
    $$ 
    where $\nabla^0$ is the corresponding GT connection. It suffices to prove that $\Psi$ is a Galilean diffeomorphism.
    
    First, $\Psi$ maps the foliation induced by $Z$ to the canonical foliation induced by $\pt$ in $I\times \f$, with $t=\pi_I \colon I \times \f \to I$. Consequently, if $\phi_q(s)$ is the flow of $Z$ for $q\in \f$, we have $$\Psi(\phi_q(s)) = (\beta(s),q'),\quad \forall s\in I\,,$$ 
    for certain $q'\in \f$ and a smooth function $\beta\in C^\infty (I)$. Since $\Psi$ preserves the causal character of vector fields for the semi-Riemannian structures, $\diff \Psi(Z)=\frac{\diff}{\ds}(\Psi \circ \phi_q)$ must be timelike unitary. Thus, $(\beta')^2 = 1$. 
    
    On the one hand, if $\beta'=-1$, we may redefine the isometry map as follows,
    $$\widetilde{\Psi}\colon M \to I \times \f, \quad p\longmapsto \widetilde{\Psi}(p) = \big(-\psi_1(p), \psi_2(p)\big)\,,$$
    for all $p\in M$, where $\psi_i$ are the component functions of $\Psi$, $\Psi(p) = \big(\psi_1(p), \psi_2(p)\big)$ with $\psi_1(p)\in I$ and $\psi_2(p)\in \f$. Hence, the orientation of the isometry $\Psi$ is reversed and the previous argument gives $\diff \widetilde{\Psi} (Z) = \widetilde{\beta'}=1$.
    
    Now, if $\beta'=1$, we have 
    $$(\Psi \circ \phi_q)(s) = (s+c,q'),\quad \forall s\in I\,,$$ 
    with $c\in I$, and $\diff \Psi(Z) = \pt$. The manifold $M$ is simply connected and $\diff \Omega =0$, then there exists a function $T\in C^\infty(M)$, determined except for a constant, such that $\Omega=\diff T$. It turns out $(t\circ \Psi)(\phi_q(s)) = s+c$  and  $T(\phi_q(s)) = s+\Tilde{c}$, with $\Tilde{c}\in \R$. If necessary, redefine $T$ via a translation to get $t\circ \Psi = T$. As a consequence, $\Psi^* \dt = \diff T$. Using $\Psi$ preserves the semi-Riemannian structures, from Theorem \ref{theor:semi-Riemannian isometry} we obtain $\Psi^* (f^2 g|_{\f}) = g$.

    Note that $\diff \Psi(Z) = \pt$ and take into account 
    $$
    \diff\Psi(P^Z X) = \diff\Psi(X) - \Omega(X) \diff\Psi(Z) = P^{\pt} \diff\Psi(X)\,,
    $$
    it is straightforward that formula 'à la Koszul' (\ref{eq:Koszul}) gives 
    \[
        2 g|_\f (P^{\pt} \nabla^0_{\diff \Psi(X)} \diff \Psi(Y), \diff\Psi(V)) = 2 g(P^Z \nabla_X Y, V)\,,
    \]
    for all $X,Y\in \Gamma(TM)$ and all $V\in \Gamma(\An)$. Therefore $\nabla$ and $\nabla^0$ are $\Psi-$related.
    %Using that $\nabla^0$ is characterized by $\operatorname{Rot}^0(\pt)= 0$ and $\nabla^0_{\pt} \pt = 0$ \cite[Theorem 27]{Bernal}, we get $\nabla$ and $\nabla^0$ are $\Psi-$related.
\end{proof}

When the spacelike leaves are compact, we can assure the existence of a  global splitting.
\begin{theor}\label{theor:splitting global compact}
    Let $\spacetime$ be a simply connected Galilean spacetime. If there exists a Galilean torqued vector field $K\in \Gamma(TM)$ and the leaves of the distribution induced by $\An$ are compact, then $M$ is a GT spacetime. 
\end{theor}
\begin{proof}
    It is enough to show that the vector field $Z=\frac{K}{\Omega(K)}$ is a uniform global generator. Indeed, let $\Phi\colon \mathcal{D} \to M$ be the maximal local flow of $Z$. 
    
    For any $p\in \f$, we have $I_p\times U_p\subset \mathcal{D}$, with $I_p\subset \R$ an open interval and $U_p$ an open set in $\f$. Since $\f$ is compact, $\f = \bigcup_{i=1}^k U_i$, $k\in \mathbb{N}$. Then, $\Phi$ is well-defined on some domain $I\times \f$, for $I\subset \R$ maximal interval. 
    
    In the case $I=\R$, the flow is complete and the proof is ended using Theorem \ref{theor:splitting global}. Suppose now $I=]a,b[$, with $b<+\infty$. Given $q\in \f$ and $\epsilon>0$, if $\Phi(\cdot, q)$ is defined on $]a,b+\epsilon[$, then there exists $\delta>0$ such that $]-\delta, \delta[ \times \f_{\Phi(b,q)}\subset \mathcal{D}$. Therefore, the flow can be extended, 
    \begin{equation*}
        \overline{\Phi}(t, q)=\left\{\begin{array}{lcc}
        \Phi(t, q) & \text { if } \quad t \in(a, b-\delta), \\
        \Phi\big(t-b+\tfrac{\delta}{2}, \Phi(b-\tfrac{\delta}{2}, q)\big) & \text { if } \quad t \in(b-\delta, b+\delta) .
        \end{array}\right.
    \end{equation*}
    This is a contradiction.
        
    Finally, we prove that $\Phi\colon I\times \f\to M$ is onto. Suppose again by contradiction, that there exists $q\in M\setminus \Phi(I\times \f)$ and consider the maximal interval $J\subset \R$ where $\Phi\colon J\times \f_q\to M$ is defined. If $\Phi(I\times \f) \cap \Phi(J\times \f_q)$ is not empty, the flow could be extended towards $I$. Therefore, since $I$ is maximal, $\Phi(I\times \f)$ is open and closed in a connected manifold $M$. Hence, $\Phi(I\times \f)=M$, and the proof is ended using Theorem \ref{theor:splitting global}.
\end{proof}

\section*{Acknowledgements}
The three authors are partially supported by Spanish MICINN
project PID2021-126217NB-I00.

\section*{Data availability}

Not applicable.

% \section{Uniqueness of Torqued Vector Fields}
% \input{uniqueness}

\bibliographystyle{plain}
\bibliography{references}

\begin{comment}

\end{comment}

\end{document}